\newif\ifarxiv
\arxivtrue

\documentclass{jgaa-art}

\usepackage{graphicx}
\usepackage{amsmath}
\usepackage{amssymb}
\usepackage{etoolbox}
\usepackage{tcolorbox}
\usepackage{zref-clever}\zcsetup{cap=true}
\usepackage{todonotes}
\usepackage{subcaption}

\newcommand{\defproblem}[3]{
  \begin{tcolorbox}[colback=white]
    \hspace{0ex}\hspace*{-2.8ex}
    \begin{minipage}{0.99\textwidth}
      \vspace{0ex}\vspace*{-1ex}
      \begin{tabular}{@{}l@{~~}p{0.9\textwidth}@{}}
        {\sf\bfseries\color{gray} Problem:} & #1\\[.1ex]
        {\sf\bfseries\color{gray} Input:} & #2\\[.1ex]
        {\sf\bfseries\color{gray} Question:} & #3\\[-1ex]
      \end{tabular}
    \end{minipage}
  \end{tcolorbox}
}

\newcommand{\PEG}{Partizan Edge Geography}
\newcommand{\DPEG}{Distinct Partizan Edge Geography}
\newcommand{\RPEGsc}{\textsc{Rooted \PEG{}}}
\newcommand{\RDPEGsc}{\textsc{Rooted \DPEG{}}}
\newcommand{\ThreeTQBF}{\textsc{3-TQBF}}

\newtheorem{lemma}[theorem]{Lemma}
\newtheorem{corollary}[theorem]{Corollary}

\AddToHook{env/lemma/begin}{%
  \zcsetup{countertype={theorem=lemma}}%
}
\AddToHook{env/corollary/begin}{%
  \zcsetup{countertype={theorem=corollary}}%
}
\AddToHook{env/observation/begin}{%
  \zcsetup{countertype={theorem=observation}}%
}

\begin{document}

\ifarxiv
    \renewcommand{\printlicense}{}
    \patchcmd{\maketitle}
    {\smallskip\hrule\smallskip\printlicense}
    {\printlicense}
    {}
    {\PackageError{main}{Could not remove the JGAA license separator}{}
    }
    \fancyhead{}
    \fancyhead[LE,RO]{\thepage}
    \pagestyle{fancy}
\else
    \doi{} %
    \Issue{0}{0}{0}{0}{0} %

    \submitted{}%
    \accepted{}%
    \published{}%
    \type{}%
    \editor{}%
\fi

\HeadingAuthor{} %
\HeadingTitle{} %
\title{The Complexity of Undirected Partizan Edge Geography} %
\Ack{Supported by JSPS KAKENHI Grant Number JP26KJ1299 and JST SPRING, Grant Number JPMJSP2125.} %

\authorOrcid[nagoya,jsps]{Yuto Okada}{research@yutookada.com}{0000-0002-1156-0383} %
\affiliation[nagoya]{Nagoya University, Japan} %
\affiliation[jsps]{JSPS Research Fellow} %

\maketitle

\begin{abstract}
    \PEG{} is a two-player game on a graph where each player has their own token on a vertex and moves their token to a neighbor in a turn removing the edge.
    Two player alternately move their tokens and the first player who cannot move loses the game.
    Fraenkel and Simonson (TCS, 1993) showed that the winner determination of this game is PSPACE-complete on directed graphs, given a graph and token positions.

    This paper resolves its complexity on undirected graphs by showing the PSPACE-completeness on bipartite undirected graphs of maximum degree 3.
    The same reduction also works for a variant where two tokens cannot be placed on the same vertex.
\end{abstract}

\section{Introduction}

(Generalized Vertex) Geography is one of the most well-studied two-player combinatorial games.
This game is played on a directed graph together with a shared token placed on some vertex.
In each turn, a player chooses an outneighbor of the vertex with the token, moves the token to it, and removes the previous vertex.
Two players alternately do this and the first player who cannot move loses the game.
The problem of, given the current graph and the token position, determining the winner is well known to be PSPACE-complete, even on planar bipartite (directed) graphs of maximum degree 3~\cite{DBLP:journals/jacm/LichtensteinS80}.
Since its use for Go~\cite{DBLP:journals/jacm/LichtensteinS80}, this result has served as one of the standard base problems for proving intractability of two-player combinatorial games~\cite{DBLP:journals/tcs/BonnetJS16,DBLP:books/daglib/0023750}.

Due to its importance, many variants of Geography are studied in the literature.
For most of them, the winner determination problems are PSPACE-complete in general with only a few exceptions; see~\cite{DBLP:journals/dam/FoxG22} for details.

This paper studies the winner determination problem of one of those variants called \PEG{}.
In this game, each player has their own token placed on a vertex (Partizan) and moves their token along an edge, removing the used edge (Edge).
Formally, in each turn, a player (1) chooses a neighbor (or an outneighbor if the graph is directed) of the vertex with their token, (2) moves the token to the neighbor, and (3) removes the used edge.
Again, two players alternately do this and the first player who cannot move loses the game.
Let us use colors \emph{blue} and \emph{red} to distinguish the two players and their tokens.
The problem is then formally defined as follows.
Note that the prefix \textsc{Rooted} is added to clarify that the starting token positions are specified; there are \textsc{Unrooted} variants, where two players can choose their starting positions first~\cite{DBLP:journals/dam/BuchananCCHCSW26}.
\defproblem
{\RPEGsc{}}
{A graph $G$ and the positions $b, r \in V(G)$ of the blue and red tokens, respectively.}
{When the blue player moves next and both players move optimally, will the blue player win the game?}

This problem was already considered in 1993 by Fraenkel and Simonson~\cite{DBLP:journals/tcs/FraenkelS93}\footnote{In their paper, the problem is called Partizan Arc Geography (PAG).} mainly on directed graphs.
They showed that the problem is PSPACE-complete on directed bipartite graphs with maximum indegree 2, maximum outdegree 2, and maximum degree 3.
Meanwhile, they gave polynomial-time algorithms for the problem on directed and undirected trees.
They also showed that the winner determination (technically, testing if the second player wins) is NP-hard on planar directed graphs and directed acyclic graphs.

\paragraph{Our results.}

The main contribution of this paper is to reveal the complexity of \RPEGsc{} on \emph{undirected} graphs, which can be seen as the counterpart on undirected graphs of the result on directed graphs~\cite{DBLP:journals/tcs/FraenkelS93}.

\begin{theorem}\label{thm:main}
    \RPEGsc{} is PSPACE-complete on bipartite undirected graphs of maximum degree 3.
\end{theorem}

We also consider a variant of \PEG{} where two tokens cannot be placed on the same vertex.
Let us refer to this variant as \emph{\DPEG{}} and the corresponding problem as \RDPEGsc{}.
With the same reduction as for \zcref{thm:main}, we also obtain the following for this problem.

\begin{corollary}\label{cor:variant}
    \RDPEGsc{} is PSPACE-complete on bipartite undirected graphs of maximum degree 3.
\end{corollary}

\paragraph{Related results.}

Fraenkel and Simonson~\cite{DBLP:journals/tcs/FraenkelS93} showed in the same paper that \textsc{Rooted Partizan (Vertex) Geography} is PSPACE-complete on \emph{directed} graphs with the same properties.
Its PSPACE-completeness on \emph{undirected} (bipartite) graphs was later shown by Fox and Geissler~\cite{DBLP:journals/dam/FoxG22}.

\section{PSPACE-Completeness}

In this section, we prove \zcref{thm:main} and confirm later that the reduction also gives \zcref{cor:variant}.
The PSPACE-membership of \RPEGsc{} is clear; observe that each state can be encoded as a tuple of a subgraph of $G$ and two vertices, that each turn an edge is removed, and that there cannot be a loop of states.

To show the PSPACE-hardness we use the following problem \ThreeTQBF{}, where each $Q_i$ is the quantifier $\exists$ when $i$ is odd and $\forall$ otherwise.
\defproblem
{\ThreeTQBF}
{Integers $n, m$ and a 3-CNF formula $\phi$ on $n$ variables $x_1, x_2, \dots, x_n$ with $m$ clauses.}
{Is $Q_1 x_1 Q_2 x_2 \dots Q_n x_n \phi(x_1, x_2, \dots, x_n)$ true?}
This problem is PSPACE-complete~\cite{DBLP:conf/stoc/StockmeyerM73} and remains so even if $n$ is odd, as we can just add a dummy variable and a dummy clause if $n$ is even (also observed in~\cite{DBLP:journals/tcs/FraenkelS93}).
Hence, in the following we assume $n$ to be odd.
We also assume that $n$ and $m$ are sufficiently large.

\paragraph{Construction.}

Let $\langle n, m, \phi \rangle$ be an instance of \ThreeTQBF{} and $\ell = 2nm$.
Before stating the whole construction, let us introduce some gadgets needed.

\begin{figure}[htbp]
    \begin{subfigure}[b]{.47\textwidth}
        \centering
        \includegraphics[page=1]{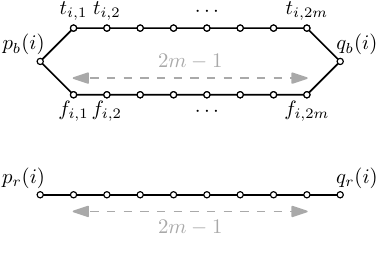}
        \subcaption{when $i$ is odd.}\label{fig:variable-gadget:odd}
    \end{subfigure}
    \hfill
    \begin{subfigure}[b]{.47\textwidth}
        \centering
        \includegraphics[page=2]{figure}
        \subcaption{when $i$ is even.}\label{fig:variable-gadget:even}
    \end{subfigure}
    \caption{The variable gadget for variable $x_i$.}
    \label{fig:variable-gadget}
\end{figure}
\begin{figure}[htbp]
    \centering
    \includegraphics[page=3]{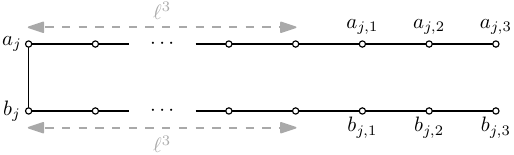}
    \caption{The clause gadget for clause $C_j$.}
    \label{fig:clause-gadget}
\end{figure}
\begin{figure}[htbp]
    \centering
    \includegraphics[page=4]{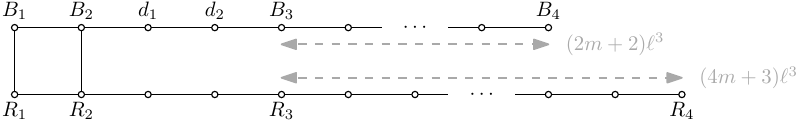}
    \caption{The judge gadget.}
    \label{fig:judge-gadget}
\end{figure}

For each variable $x_i$, we define the \emph{variable gadget} for $x_i$ as the graph depicted in \zcref{fig:variable-gadget:odd} if $i$ is odd and the graph depicted in \zcref{fig:variable-gadget:even} otherwise.
Although the two graphs are isomorphic, we name vertices in a different manner depending on the parity of $i$.
For each clause $C_j$ in $\phi$, we define the \emph{clause gadget} for $C_j$ as the graph depicted in \zcref{fig:clause-gadget}.
Lastly, we call the graph depicted in~\zcref{fig:judge-gadget} the \emph{judge gadget}.
Note that each length in the figures denotes the number of edges.

Using the above gadgets, we construct an instance $\langle G, b, r \rangle$ of \RPEGsc{} in the following manner; see also \zcref{fig:whole-construction}.
\begin{enumerate}
    \item Let $G$ be the graph consisting only of the judge gadget.
    \item For each $1 \leq i \leq n$, add the variable gadget for $x_i$ to $G$.
    \item For each $1 \leq i < n$, add two edges $\{q_b(i), p_b(i+1)\}$ and $\{q_r(i), p_r(i+1)\}$ to $G$.
    \item For each $1 \leq j \leq m$, add the clause gadget for $C_j$ to $G$.
    \item For each $1 \leq j < m$, add edge $\{b_j, a_{j+1}\}$ to $G$.
    \item Add two edges $\{q_b(n), B_1\}$ and $\{q_r(n), R_1\}$ to $G$.
    \item Add a path consisting of $\ell^3$ edges each, between $d_1$ and $a_1$ and between $b_m$ and $d_2$ to $G$.
    \item For each $1 \leq j \leq m$ and each $1 \leq k \leq 3$, let $i$ be the index of the variable used for the $k$-th literal of clause $C_j$. Unless the $k$-th literal already appears as the $k'$-th literal for some $k' < k$, we connect the clause gadget for $C_j$ with the variable gadget for $x_i$. Observe that $G$ is connected and bipartite after the last step. Note that $\ell^3 = 8n^3m^3$ is even. We adjust the path lengths according to the unique bipartition so that adding those paths does not break the bipartiteness.
          \begin{itemize}
              \item If it is a positive literal, add two paths, consisting of $\ell^2$ or $\ell^2 + 1$ edges each, between $a_{j,k}$ and $t_{i,2j}$ and between $b_{j,k}$ and $t_{i,2j-1}$ to $G$.
              \item If it is a negative literal, add two paths, consisting of $\ell^2$ or $\ell^2 + 1$ edges each, between $a_{j,k}$ and $f_{i,2j}$ and between $b_{j,k}$ and $f_{i,2j-1}$ to $G$.
          \end{itemize}
\end{enumerate}

\begin{figure}[htbp]
    \centering
    \includegraphics[page=5]{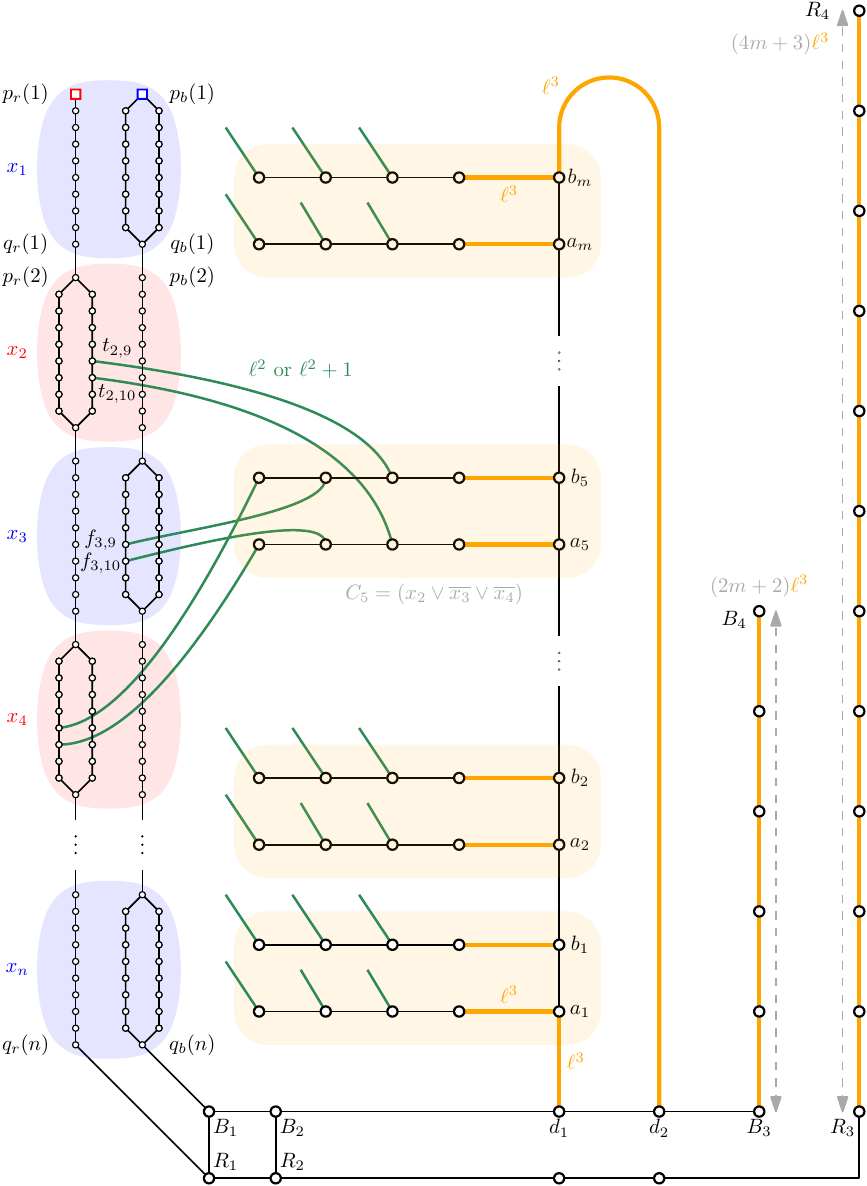}
    \caption{The constructed graph $G$.}
    \label{fig:whole-construction}
\end{figure}

Observe that $G$ obtained in this manner is bipartite and has degree 3.
Moreover, it can be constructed in polynomial time as $\ell = O(nm)$.
We let $b = p_b(1), r = p_r(1)$ and complete the construction.

Intuitively, we expect both players to go through the variable gadgets first, assigning true or false to the variables alternately.
We then expect the red player to proceed as $R_1 \rightarrow R_2 \rightarrow \dots \rightarrow R_3$.
The blue player is expected to proceed as $B_1 \rightarrow B_2 \rightarrow d_1 \rightarrow a_1 \rightarrow b_1 \rightarrow \dots \rightarrow a_m \rightarrow b_m \rightarrow d_2 \rightarrow B_3$.
For each satisfied clause, the blue player can take a detour of length approximately $2\ell^3$ through its satisfying variable.
Hence, if all the clauses are satisfied, the blue player can move approximately $(4m + 4)\ell^3$ times, and otherwise $(4m + 2)\ell^3$ times in the best case, which can be distinguished by the number of possible moves by the red player, approximately $(4m + 3)\ell^3$.
However, a player might still be able to win the game by deviating from the expected behavior.
We first confirm that this is in fact impossible.

\begin{lemma}\label{lem:unexpected-behavior}
    The first player who diverges from the following token movements loses the game.
    \begin{itemize}
        \item[(blue)] $p_b(1) \rightarrow (t_{1,1} \rightarrow \dots \rightarrow t_{1,2m} \text{ / } f_{1,1} \rightarrow \dots \rightarrow f_{1,2m}) \rightarrow q_b(1) \rightarrow p_b(2) \rightarrow \dots \rightarrow q_b(2) \rightarrow p_b(3) \rightarrow (t_{3,1} \rightarrow \dots \rightarrow t_{3,2m} \text{ / } f_{3,1} \rightarrow \dots \rightarrow f_{3,2m}) \rightarrow q_b(3) \rightarrow \dots \rightarrow q_b(n) \rightarrow B_1$.
        \item[(red)] $p_r(1) \rightarrow \dots \rightarrow q_r(1) \rightarrow p_r(2) \rightarrow (t_{2,1} \rightarrow \dots \rightarrow t_{2,2m} \text{ / } f_{2,1} \rightarrow \dots \rightarrow f_{2,2m}) \rightarrow q_r(2) \rightarrow p_r(3) \rightarrow \dots \rightarrow q_r(3) \rightarrow \dots \rightarrow q_r(n) \rightarrow R_1$.
    \end{itemize}
\end{lemma}

\begin{proof}
    Let us consider the state just after the turn on which such a divergence occurred.
    Then, the other player, say $c$, is still on track.
    We claim that player $c$ can win the game no matter how the opponent moves.
    If $c$ is the blue player, it suffices to move their token to $B_1$ along the above path and then proceed as follows:
    \begin{align*}
        B_1 \rightarrow R_1 \rightarrow R_2 \rightarrow B_2 \rightarrow d_1 \rightarrow \dots \rightarrow a_1 \rightarrow b_1 \rightarrow a_2 \rightarrow b_2 \rightarrow \dots \rightarrow a_m \rightarrow b_m \rightarrow \dots \rightarrow d_2 \rightarrow B_3.
    \end{align*}
    It is always possible for the blue player to do this movement.
    To interfere with this movement, the red player would first need to use a path between a variable gadget and a clause gadget of length at least $\ell^2$.
    This is more than enough for the blue token to reach $d_1$, which requires $O(mn)$ turns.
    This is less than $\ell^2 = 4n^2m^2$ for sufficiently large $n, m$.
    To interfere with the rest of the movement, namely from $d_1$ to $d_2$, the red player would need to use at least $\ell^3$ turns additionally to reach $a_j$ or $b_j$ for some $j$.
    Hence, the red player needs at least $\ell^2 + \ell^3$ turns after the divergence turn to do so, while the blue player already reaches $b_m$ in $\ell^3 + O(mn)$ turns, which is again less than $\ell^3 + \ell^2 = \ell^3 + 4n^2m^2$ for sufficiently large $n,m$.

    Observe now that, with the above movement of the blue player, the number of possible edges that the red player could use in the game can already be bounded from above by $2m\ell^3 + O(m\ell^2)$:
    due to the blue player moving their token as $q_b(n) \rightarrow B_1 \rightarrow R_1 \rightarrow R_2$, $d_1 \rightarrow \dots \rightarrow a_1$, and $b_m \rightarrow \dots \rightarrow d_2$, the red token is essentially trapped in the variable gadgets, the clause gadgets, and the paths between them.
    Meanwhile, the blue player can still make $(2m+2)\ell^3 = 2m\ell^3 + \Theta(nm\ell^2)$ moves and win the game for sufficiently large $n, m$.

    If $c$ is the red player, it suffices to move their token to $R_1$ along the above path and then proceed as follows:
    \begin{align*}
        R_1 \rightarrow B_1 \rightarrow B_2 \rightarrow d_1 \rightarrow \dots \rightarrow a_1 \rightarrow b_1 \rightarrow a_2 \rightarrow b_2 \rightarrow \dots \rightarrow a_m \rightarrow b_m \rightarrow \dots \rightarrow d_2 \rightarrow B_3.
    \end{align*}
    Almost the same arguments can be applied to show that the red player wins the game by doing this movement.
\end{proof}

By \zcref{lem:unexpected-behavior}, we can assume that both players move their tokens as stated in \zcref{lem:unexpected-behavior}.
With this assumption we show the correctness of the reduction in the following.

\begin{lemma}\label{lem:correctness}
    The blue player wins if and only if $Q_1 x_1 Q_2 x_2 \dots Q_n x_n \phi(x_1, x_2, \dots, x_n)$ is true.
\end{lemma}

\begin{proof}
    For each $i$, let us regard passing through $f_{i, 1}, f_{i, 2}, \dots, f_{i, 2m}$ inside the variable gadget $i$ as assigning true to variable $x_i$, and passing through $t_{i, 1}, t_{i, 2}, \dots, t_{i, 2m}$ as assigning false to $x_i$.

    Suppose that $Q_1 x_1 Q_2 x_2 \dots Q_n x_n \phi(x_1, x_2, \dots, x_n)$ is true.
    Then, the blue player can win the game just by moving their token so that $\phi(x_1, x_2, \dots, x_n)$ becomes true, which is always possible no matter how the red token moves, due to $Q_1 x_1 Q_2 x_2 \dots Q_n x_n \phi(x_1, x_2, \dots, x_n)$ being true.
    Observe that, just after the turn on which the red token is placed on $R_1$, the blue token is placed on $B_1$.
    When both players play optimally, the blue token moves as $B_1 \rightarrow B_2 \rightarrow d_1$ and the red token moves as $R_1 \rightarrow R_2 \rightarrow \dots \rightarrow R_3$, since otherwise the first player to violate this would immediately lose the game.
    From $R_2$, the red player can spend $(4m + 3) \ell^3 + O(1)$ turns.
    Meanwhile, the blue player can spend at least $(4m + 4) \ell^3$ turns by moving their token as follows.
    \begin{enumerate}
        \item Move as $d_1 \rightarrow \dots \rightarrow a_1$.
        \item For each $1 \leq j \leq m$ (in increasing order), let $i, k$ be integers such that the $k$-th literal with variable $x_i$ satisfies clause $C_j$ and move as $a_j \rightarrow \dots \rightarrow a_{j,k} \rightarrow \dots \rightarrow (t / f)_{i,2j} \rightarrow (t / f)_{i,2j-1} \rightarrow \dots \rightarrow b_{j,k} \rightarrow \dots \rightarrow b_j$.
        \item Move as $b_m \rightarrow \dots \rightarrow d_2 \rightarrow \dots \rightarrow B_3 \rightarrow \cdots$.
    \end{enumerate}
    Note that this movement is possible by the construction of $G$ and the fact that the red player cannot go back after heading to $R_3$ from $R_2$.
    Hence, the blue player can win the game whenever $Q_1 x_1 Q_2 x_2 \dots Q_n x_n \phi(x_1, x_2, \dots, x_n)$ is true.

    \begin{figure}[htbp]
        \centering
        \includegraphics[page=6]{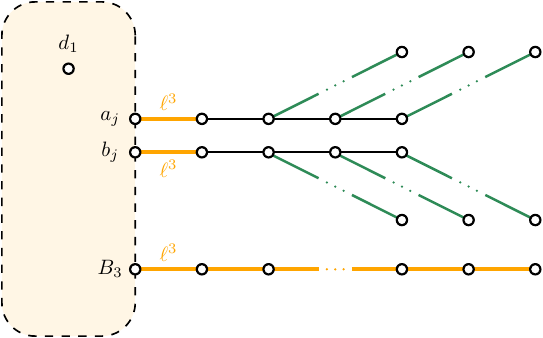}
        \caption{How the component that the blue token belongs to looks.}
        \label{fig:three-subtrees}
    \end{figure}

    Suppose that $Q_1 x_1 Q_2 x_2 \dots Q_n x_n \phi(x_1, x_2, \dots, x_n)$ is false.
    In this case, the red player can win the game just by moving their token so that $\phi(x_1, x_2, \dots, x_n)$ becomes false.
    This is again always possible no matter how the blue player moves.
    By doing so, when the red token arrives at $R_1$, there is at least one clause $C_j$ such that the variable-gadget-side endpoints of paths of length $\ell^2$ or $\ell^2 + 1$ attached to the clause gadget for $C_j$ are all leaves.
    Hence, from the viewpoint of the blue token placed on $d_1$, the component containing $d_1$ can be viewed as depicted in \zcref{fig:three-subtrees}.
    There are three subtrees rooted at $a_j$, $b_j$, and $B_3$ each containing a path of length $\ell^3$.
    As $d_1$ is not in those subtrees, starting from $d_1$, at least two of those three paths of length $\ell^3$ cannot be traversed by the blue token.
    This results in an upper bound of $(4m+2) \ell^3 + O(m \ell^2) = (4m + 2 + o(1)) \ell^3$ on the number of possible edges that the blue player could use after placing their token on $B_2$.
    This is less than the number of moves available to the red player from $R_2$ for sufficiently large $n, m$, and hence, the red player wins.
\end{proof}

This completes the proof of \zcref{thm:main}.
Lastly, let us briefly confirm that the same reduction yields \zcref{cor:variant}, where two tokens are not allowed to be placed on the same vertex.
In \zcref{lem:unexpected-behavior}, it is clear that there is no chance that the two tokens are placed on the same vertex or on adjacent vertices throughout the winning strategy for the other player $c$.
Hence, we can assume that both players play the game as stated in \zcref{lem:unexpected-behavior}.
After such play, the two tokens become adjacent, occupying vertices $B_1$ and $R_1$.
However, there is no problem, as even with the additional rule the optimal movements are to move the blue token as $B_1 \rightarrow B_2 \rightarrow d_1$ and the red token as $R_1 \rightarrow R_2 \rightarrow \dots \rightarrow R_3$, which is exactly what we expected in the proof of \zcref{lem:correctness}.
Since there is no further possibility that the two tokens become adjacent, \zcref{cor:variant} can be shown by almost the same arguments.

\section{Conclusions}

In this paper, we showed the PSPACE-completeness of \RPEGsc{} and its variant \RDPEGsc{}.
Following these results, we conclude this paper by posing some open problems.
\begin{itemize}
    \item \textsc{Geography} is known to be PSPACE-complete on planar directed graphs with maximum degree 3. We conjecture that the two problems in this paper remains PSPACE-complete on such graphs.
    \item We conjecture that the \emph{\textsc{Unrooted}} variants, where the players can choose the starting vertices before the game, are also PSPACE-complete.
    A variant of Unrooted \PEG{} called \emph{Trail Trap} is known to be NP-hard~\cite{DBLP:journals/dam/BuchananCCHCSW26} to test if the second player wins and the authors conjectured that the problem is PSPACE-complete~\cite[Conjecture 23]{DBLP:journals/dam/BuchananCCHCSW26}.
    We believe that these three \textsc{Unrooted} problems are so similar that showing the PSPACE-completeness of any one of them would give sufficient insight for resolving the others.
\end{itemize}

\section*{AI Usage Statement}

The (original) proofs in this paper were obtained by the author without any use of generative AI. However, the author used GPT-5.6 Luna, GPT-5.6 Sol, and GPT-6 Astra for searching related literature, checking the proofs, and improving the writing.

\section*{Acknowledgements}

I thank Kanae Yoshiwatari for introducing the problem to me and checking an earlier draft.

\bibliography{main}
\bibliographystyle{abbrvurl}

\end{document}